\documentclass[final,twoside,11pt]{entics} 
\usepackage{enticsmacro}

\usepackage{rotating}
\usepackage{amssymb}
\usepackage{graphicx}
\usepackage{url}
\usepackage{stmaryrd}
\usepackage{amsmath,amssymb, amsfonts}
\usepackage{bussproofs}
\usepackage{graphicx}
\usepackage{wasysym}
\input{xy}
\xyoption{all}
\def\[{\ensuremath{[ \! [}}
\def\]{\ensuremath{] \! ]}}

\newcommand{\topp}{{\mathsf{top}}}
\newcommand{\leqt}{<\!:}
\newcommand{\fsub}{${\mathsf{F_{\leqt}}}\hspace{1pt}$}
\newcommand{\fwedge}{${\mathsf{F_\wedge}}\hspace{1pt}$}
\newcommand{\fsubtop}{${\mathsf{F_{\leqt}^\top}}\hspace{1pt}$}
\newcommand{\fsubu}{${\mathsf{F_{\leqt}^{K\top}}}\hspace{1pt}$}
\def\k{{\mathsf{K}}}

\def\Y{{\mathbf{Y}}}

\def\0{{\mathbf{0}}}

\def\1m{\iota}

\def\1{{{\mathbf{1}}}}

\newcommand{\dom}{{\mathsf{dom}}}

\newcommand{\spc}{\hspace{1pt}}

\def\conf{MFPS 2022} 	
\volume{1}			
\def\volu{1}			
\def\papno{10}			
\def\mydoi{{\normalshape \href{https://doi.org/10.46298/entics.10474}{10.46298/entics.10474}}}
\def\copynames{J. Laird} 
\def\runtitle{Revisiting Decidable Bounded Quantification, via Dinaturality}
\def\lastname{Laird}
\begin{document}
\begin{frontmatter}
\title{Revisiting Decidable Bounded Quantification,\\ via Dinaturality}
\author{J. Laird}
\address{Department of Computer Science, University of Bath, UK}


\begin{abstract}
We use a semantic interpretation to investigate the problem of defining an expressive but decidable type system with bounded quantification. Typechecking in the widely studied System \fsub is undecidable, thanks to an undecidable subtyping relation, for which the culprit is the rule for subtyping bounded quantification. Weaker versions of this rule, allowing decidable subtyping, have been proposed. One of the resulting type systems (Kernel \fsub\!) lacks expressiveness, another (System \fsubtop\!) lacks the minimal typing property and thus has no evident typechecking algorithm.

We consider these rules as defining distinct forms of bounded quantification,  one for interpreting type variable abstraction, and the other for type instantiation. 
 By giving a semantic interpretation for both in terms of unbounded quantification, using the dinaturality of type instantiation with respect to subsumption, we show that they can coexist within a single type system. This does have the minimal typing property and thus a simple typechecking procedure.

We consider the fragments of this unified type system over types which contain only one form of bounded quantifier. One of these is Kernel \fsub\! while the other can type strictly more terms than System \fsubtop but the same set of $\beta$-normal terms. We show decidability of typechecking for this fragment, and thus for System \fsubtop typechecking of $\beta$-normal terms.
\end{abstract}
\end{frontmatter}

\section{Introduction}
By combining subtype and parametric polymorphism, type systems with bounded quantification may be used to write programs which  are generic, but range over a constrained set of types (a program of  type $\forall X \leqt S.T$ may be instantiated only with a subtype of $S$).  They  have been used to develop theories of  key aspects of object oriented  languages such as  \emph{inheritance}  \cite{CFun}.  However, the problem of designing a tractable but expressive type system with bounded quantification is surprisingly difficult. The most natural and widely studied system based on the $\lambda$-calculus --- System \fsub --- has an undecidable subtyping relation \cite{Pierce1}, and thus an undecidable typing relation. It has nonetheless been influential in the development of subsequent type systems with bounded quantification such as the DOT (Dependent Object Types) calculus \cite{dot1,dot2}. Undecidability of subtyping for a key fragment of this system has been shown by reduction to  System \fsub \cite{dot1,HuL}.

Attempts to modify System \fsub to recover decidability of subtyping have been partially successful: amongst these, returning to the weaker subtyping rule  for bounded quantification from the Fun calculus \cite{CFun} ($\mathsf{\forall-Fun}$) gives a well-behaved system (Kernel \fsub\!) with decidable typechecking, at the cost of a rather arbitrary restriction on the  subtyping relation, leaving it unable to express some natural, and potentially useful,  instances of subtyping. These \emph{are} captured by System \fsubtop\!, a version  of System \fsub  with a different subtyping rule for bounded quantification ($\mathsf{\forall -Top}$) leading to a decidable and reasonably expressive subtyping relation but to date no effective, sound and complete typechecking procedure \cite{fsubtop}. A third  potential replacement for the quantifier subtyping rule  ($\mathsf{\forall-Loc}$) is intuitively appealing but does not have an evident subtyping algorithm  \cite{fsubtop}.
Other decidable variants include the (heavy) restriction of not allowing quantification bounds to include the $\top$ type \cite{KaS}, and a family of structural \emph{extensions} of System \fsub with decidable subtyping \cite{vor}.  
  More recently,  Strong Kernel \fsub has benn proposed \cite{HuL} as a decidable subtyping system for \fsub types using two contexts of type bounds, which is strictly more expressive than Kernel \fsub\!.

We revisit these problems from a semantic perspective, returning to the original, framework (single contexts, no restriction on which types can appear as bounds). We consider the two different rules from Kernel \fsub\ \! and System \fsubtop\! for subtyping bounded quantification as defining different forms of bounded  quantifier  --- one for typing abstraction of type variables, and one for typing instantiation,
 within a single type system: System \fsubu\!.

Developing  an interpretation of bounded types proposed   in \cite{piercethesis}, we give interpretations of these two bounded quantifiers   in terms of unbounded quantification and a meet operation on types, and derive a version of the  $\mathsf{\forall-Loc}$ rule for inferring the subtyping relation between them. 
We extend this interpretation to a type system for System \fsub terms, and show soundess with respect to second-order $\beta$ and $\eta$ equality. This depends on the \emph{dinaturality} property of quantifier instantiation with respect to subsumption (which was  introduced to the equational theory of System \fsub by Cardelli et. al. \cite{fsub}).  

The subtyping and typechecking algorithms for System \fsub \cite{CGhe} adapt readily to our unified type system. Its most practically relevant fragments are those in which the quantifiers in the types annotating terms and contexts all satisfy the same subtyping rule. One is simply Kernel \fsub itself. The other is a modest extension of System  \fsubtop to include terms typable with the ``missing'' minimal types to allow a simple, terminating typechecking procedure. We establish that this fragment is nonetheless semantically equivalent to System \fsubtop\!, by showing that it types the same $\beta$-normal terms, for which  System \fsubtop\! typechecking is therefore decidable.

\section{Background: Subtyping Bounded Quantification}
We first review System \fsub \cite{CGhe,fsub} (and its subsystems,  with their subtyping and typechecking problems). Its raw types  are given by the grammar:
 $$T::= \top\ |\ X\ |\ T \rightarrow T\ |\ \forall (X \leqt T).T$$
where $X$ ranges over a set of type-variables. We follow convention in defining the unbounded quantification $\forall X.S$ to be $\forall (X \leqt \top).S$ and identifying types up to $\alpha$-conversion of bound variables.  

A \emph{context} $\Theta$ is a sequence of assumptions of the form $X \leqt T$ (the type-variable X has bound $T$) or $x:T$ (the term-variable $x$ has type $T$). Judgments $\Theta \vdash T$   --- ``$T$ is a well-defined type in the well-defined context $\Theta$'' --- are derived according to the rules in Table \ref{tfr}. ($\Theta \vdash \top$ thus means that $\Theta$ is a well-defined context.)  
\begin{table}
\begin{center}
\AxiomC{${\color{white} T}$} 
\UnaryInfC{\small $ \_ \vdash \top$} 
\DisplayProof\ \ \ 
 \AxiomC{\small$\Theta \vdash S$} 
\AxiomC{\small$ \Theta \vdash T $} 
\RightLabel{\small $X \not \in \dom(\Theta)$}
\BinaryInfC{\small $\Theta, X \leqt S \vdash T $}  
\DisplayProof\ \ \ 
\AxiomC{\small $\Theta \vdash S$}
\AxiomC{\small $\Theta \vdash T$}
\RightLabel{\small $x \not \in \dom(\Theta)$}
\BinaryInfC{\small$ \Theta, x:S \vdash T $}
\DisplayProof\end{center}
\begin{center}
\AxiomC{\small $\Theta,X \leqt T \vdash \top$}\UnaryInfC{\small$ \Theta,X\leqt T \vdash X$}\DisplayProof\ \ \ 
\AxiomC{\small $  \Theta \vdash S$}\AxiomC{ \small $ \Theta \vdash T$}\BinaryInfC{\small$ \Theta \vdash S \rightarrow T $}\DisplayProof \ \ \  
\AxiomC{\small $\Theta, X \leqt S \vdash T$}\UnaryInfC{\small$  \Theta \vdash \forall (X \leqt S).T$}  \DisplayProof 
\end{center}
\begin{center}
\caption{Type and Context Formation Rules for System \fsub}   \label{tfr}
\end{center}
\end{table}

\begin{table}
\begin{center}
 \AxiomC{\small $\Theta, X \leqt T,\Theta' \vdash \top$}\RightLabel{\small $\mathsf{Var}$} 
\UnaryInfC{\small $\Theta, X \leqt T,\Theta' \vdash X \leqt T $} \DisplayProof\ \ \ \AxiomC{\small $\Theta \vdash T$}\RightLabel{\small $\mathsf{Top}$} \UnaryInfC{\small $\Theta \vdash T \leqt \top$}\DisplayProof\ \ \  \AxiomC{\small $\Theta \vdash T$}\RightLabel{\small $\mathsf{Refl}$}  \UnaryInfC{\small $\Theta \vdash T \leqt T$}\DisplayProof\ \ \
  \AxiomC{\small $\Theta \vdash T \leqt T'$} \AxiomC{\small $  \Theta \vdash T' \leqt T''$} \RightLabel{\small $\mathsf{Trans}$} \BinaryInfC{\small$\Theta \vdash T \leqt T''$}\DisplayProof
\end{center}
\begin{center}
  \AxiomC{\small $\Theta \vdash S' \leqt S$}\AxiomC{\small$  \Theta \vdash T \leqt T'$}\RightLabel{\small $\mathsf{\rightarrow}$}  \BinaryInfC{\small $ \Theta \vdash S \rightarrow T \leqt S' \rightarrow T'$}\DisplayProof \ \ \ 
\AxiomC{\small $\Theta \vdash T_0 \leqt S_0$}  \AxiomC{\small $ \Theta, X \leqt T_0 \vdash S_1 \leqt T_1$}\RightLabel{\small $\mathsf{\forall-Orig}$} \BinaryInfC{\small $\Theta \vdash \forall (X \leqt S_0).S_1 \leqt \forall(X \leqt T_0).T_1$}
\DisplayProof
\end{center}
\begin{center}
\caption{Subtyping Rules for System \fsub}   \label{stj}
\end{center}
\end{table}

Subtyping judgments --- $\Theta \vdash S \leqt T$ (where $\Theta \vdash S$ and $\Theta \vdash T$) --- are derived according to the rules in Table \ref{stj}. This includes the following rule for subtyping bounded quantifications:
 \begin{center}
\AxiomC{$\Theta \vdash T_0 \leqt S_0 $}\AxiomC{$\Theta, X\leqt T_0 \vdash S_1 \leqt T_1$}\RightLabel{$\mathsf{\forall-Orig}$}\BinaryInfC{$\Theta \vdash  \forall (X \leqt S_0). S_1 \leqt  \forall (X \leqt T_0). T_1$}  \DisplayProof
\end{center}
Although well-motivated semantically, this rule is problematic from an algorithmic point of view. Reading from the bottom up, the bound on instances of  $X$ occuring in $S_1$ changes from $S_0$ to $T_0$. This ``re-bounding'' prevents the sound and complete subtyping algorithm (deterministic search procedure  for subtyping derivations) given in \cite{CGhe} from terminating on all inputs and indeed permits the encoding of a two-counter machine as a subtyping problem \cite{Pierce1}. The subtyping relation for System \fsub  is therefore undecidable.

Various proposals have been made to describe a more tractable subtyping relation on System \fsub types, including ({\it{inter alia}}) more restricted forms of the quantifier subtyping  rule, which we now describe. Instead of viewing these variants as simply defining different type systems, we will also treat them as defining different quantifiers (potentially within a single type system), which we distinguish by decorating them with different superscripts.

{\bf{Proposal 1}}  Restrict the subtype order on quantified types to those which have equal bounds. This corresponds to the rule ($\mathsf{\forall-Fun}$) for subtyping bounded quantification from Cardelli and Wegner's original \emph{Fun} calculus \cite{CFun}:
\begin{center}

\begin{prooftree}
\AxiomC{$\Theta, X\leqt S \vdash T \leqt T'$}\RightLabel{$\mathsf{\forall-Fun}$}
 \UnaryInfC{$\Theta \vdash  \forall^\k (X \leqt S).T \leqt \forall^\k (X \leqt S).T'$}
\end{prooftree}
\end{center}
This yields  a type system --- \emph{Kernel \fsub} ---  which is algorithmically well-behaved (subtyping and type-checking are efficiently decidable) at a significant cost in expressiveness: quantified types may only be compared if their bounds are the same. For example, in System \fsub an unbounded quantification $\forall X.T$ is always a subtype of any  bounded quantification  $\forall (X \leqt S).T$ with the same body, whereas in Kernel \fsub it may only be a subtype of another unbounded quantification.
 Another  example: if $S \leqt S'$ then the  abstract data type  $\exists (X \leqt S').T$ is ``more abstract'' than $\exists (X \leqt S).T$ in the sense that it is less constrained in the types that may be used to implement $X$. Representing $\exists (X \leqt S).T$ as $\forall Y. (\forall (X \leqt S).(T \rightarrow Y))\rightarrow Y$, this is captured as a subtyping  in the original system (i.e. $\exists (X \leqt S).T \leqt \exists  (X \leqt S').T$) but not in Kernel \fsub ---  $\exists^\k (X \leqt S).T \not\leqt \exists^\k  (X \leqt S').T$ if $S \not \equiv S'$.

 {\bf{Proposal 2}}  Ignore the bounds on the quantified variable when inferring the subtype relation on the bodies of quantified types. This corresponds to the following  rule ($\mathsf{\forall-Top}$), which  was proposed by Castagna and Pierce  as the basis of System \fsubtop \cite{fsubtop}: 
\begin{center}
\begin{prooftree}
\AxiomC{$\Theta \vdash  T_0 \leqt S_0$} \AxiomC{$ \Theta,X \leqt \top \vdash S_1 \leqt T_1$}\RightLabel{$\mathsf{\forall-Top}$}  
\BinaryInfC{$\forall^\top (X \leqt S_0).S_1 \leqt \forall^\top (X \leqt T_0).T_1$
}
\end{prooftree}
\end{center}
While not strictly more expressive than Kernel \fsub\!, this yields an expressive subtyping relation (e.g. capturing the relative abstractness of existential types) which has useful properties, including  decidability and the existence of meets and joins for bounded types. However, it does not interact nicely with the typing rules of  System \fsub\!  for reasons we shall now discuss,

\subsection{Typing Bounded Quantification}
The raw terms of System \fsub are given by the grammar:
$$t::= \topp\ |\ x\ |\ \lambda (x:T).t\ |\ \Lambda (X \leqt T).t\ |\  t\spc t\ |\ t\{T\}$$  
where $x$ ranges over the term variables, $X$ over the type variables, and  $T$ over the raw types of  System \fsub\!. We write $\Lambda X.t$ for $\Lambda (X \leqt \top).t$ and identify terms up to $\alpha$-conversion.

A term-in-context (or just a term, for short) $\Theta \vdash t$ consists of a well-formed context $\Theta$  and a raw term $t$. Typing judgments $\Theta \vdash t:T$, which are derived according to the rules in Table \ref{tj}, associate a term-in-context $\Theta \vdash t$ to a well-formed type  over the same context $\Theta \vdash T$. 
Thanks to the typing rule of \emph{subsumption} ($\mathsf{sub}$), the problem of determining whether a given typing judgment is derivable (typechecking) depends on the subtyping problem: e.g. for any types $\Theta \vdash S,S',T$,  the term $\Theta,x:S, f:S'\rightarrow T \vdash f\spc x$ is typable (with $T$) if and only if $\Theta \vdash S \leqt S'$. 

The (sound and complete) typechecking procedure for System \fsub (which is described in more detail in Section \ref{subt}) is based on finding a  \emph{minimal type}  for each typable term.
\begin{definition}$\Theta \vdash T$ is a minimal type for the term $\Theta \vdash t$  if   $\Theta \vdash t:T$, and if   $\Theta \vdash t:T'$ then $\Theta \vdash T \leqt T'$.
\end{definition}
System \fsub  possesses the \emph{minimal typing property} with respect to the original subtyping rule: every term which can be typed has a minimal type.     Thus we may check the typing $\Theta \vdash t:T$ by finding a minimal type $\Theta \vdash S$ for   $\Theta \vdash t$ (for which there is  a sound and complete algorithm) and checking $\Theta \vdash S \leqt T$.  The minimal typing algorithm for System \fsub restricts straightforwardly to Kernel \fsub\!. Since the subtyping problem for Kernel \fsub is decidable, this gives an efficient typechecking procedure which terminates on all inputs. However:
\begin{proposition}\label{ghelli} System \fsubtop does not possess the minimal typing property.
\end{proposition}
\begin{proof}
The following counterexample was given by Ghelli (\cite{fsubtop} --- Appendix).\\
 In System \fsubtop\!, the typing judgments $X \leqt \top \vdash \Lambda (Z \leqt X).\lambda (y:Z).y:\forall^\top (Z \leqt X). (Z \rightarrow  Z)$ and $X \leqt \top \vdash \Lambda (Z \leqt X).\lambda (y:Z).y:\forall^\top(Z \leqt X).(Z \rightarrow  X)$ are both derivable, but these types have no lower bound and so $X \leqt \top \vdash \Lambda (Z \leqt X).\lambda (y:Z).y$ can have no minimal type.
 
Informally, it is easy to see that such a lower bound would have the form $X \leqt \top \vdash \forall^\top (Z \leqt S). T \rightarrow T'$ where $X \leqt \top,Z \leqt \top \vdash T' \leqt X$ and $X \leqt \top,Z \leqt \top \vdash T' \leqt Z$, and that no such type exists. A more formal proof can be given using an alternative, algorithmic presentation of the subtyping relation as in Table \ref{suba}.   
\end{proof}
The problem is that the subtyping assumption $Z \leqt X$ may be used to derive the \emph{typing} $X \leqt \top \vdash \Lambda (Z \leqt X).\lambda (y:Z).y:\forall^\top(Z \leqt X).(Z \rightarrow  X)$ but not the \emph{subtyping} $X \leqt \top \vdash \forall^\top (Z \leqt X). (Z \rightarrow  Z) \leqt \forall^\top (Z \leqt X). (Z \rightarrow  X)$.   
It is not known whether the type synthesis algorithm for System \fsub can be adapted to System \fsubtop\!, nor indeed whether its typechecking problem is decidable.

\begin{table}
\begin{center}
\AxiomC{\small$\Theta \vdash \top$}\RightLabel{\small $\mathsf{top}$}\UnaryInfC{\small $\Theta \vdash \topp:\top$}\DisplayProof\ \  
\AxiomC{\small $\Theta,x:T,\Theta' \vdash \top$} \RightLabel{\small $\mathsf{var}$}\UnaryInfC{\small $\Theta,x:T,\Theta' \vdash x:T$}\DisplayProof\ \   \AxiomC{\small $\Theta \vdash t:T$} \AxiomC{\small$\Theta \vdash T \leqt T'$}\RightLabel{\small $\mathsf{sub}$}\BinaryInfC{\small $\Theta  \vdash t:T'$}\DisplayProof
\end{center}
\begin{center}
\AxiomC{\small $\Theta,x:S \vdash  t:T$}\RightLabel{\small $\mathsf{\rightarrow - i}$} \UnaryInfC{\small $  \Theta\vdash \lambda x:S.t:S \rightarrow T$}\DisplayProof \ \  \AxiomC{\small $\Theta \vdash  t:S \rightarrow T$} \AxiomC{\small $\Theta \vdash s:S$} \RightLabel{\small $\mathsf{\rightarrow - e}$}  \BinaryInfC{\small $\Theta \vdash t\spc s:T$} \DisplayProof
\end{center}
\begin{center}
\AxiomC{\small $\Theta, X \leqt S \vdash t:T$}\RightLabel{\small $\mathsf{\forall - i}$} \UnaryInfC{\small$ \Theta \vdash \Lambda (X \leqt S).t:\forall (X \leqt S).T$}\DisplayProof \ \  \AxiomC{\small $\Theta \vdash t:\forall (X \leqt S).T$} \AxiomC{\small $  \Theta \vdash S' \leqt S $} \RightLabel{\small $\mathsf{\forall - e}$} \BinaryInfC{\small $\Theta \vdash t\{S'\}:T[S'/X]$}\DisplayProof
\end{center}

\begin{center}
\caption{Typing Judgments for System \fsub}\label{tj}
\end{center}
\end{table}

\section{Semantics of Bounded  Quantification}
Following the suggestion \cite{fsubtop} that an expressive yet tractable type system for bounded quantification should be grounded in semantic understanding, we seek an interpretation which relates the bounded quantifiers $\forall^\top$ and $\forall^\k$. Since System \fsubtop and Kernel \fsub are subsystems of the original System \fsub\!, they may both be interpreted in any semantics of the latter, examples of which include models based on modest sets and partial equivalence relations \cite{BruceLongo}, and on games and strategies \cite{Chrob,fos16}. However, because  these  interpret the original subtyping rule for bounded quantification they unsurprisingly  give few direct clues about decidable subtyping: the PER models are essentially realizability interpretations on an untyped model, whereas the game semantics in \cite{fos16} interprets subtyping coercions as morphisms defined inductively on the derivation of the subtyping relation. On the other hand, the latter interpretation does depend crucially on the \emph{dinaturality} of instantiation with respect to subtyping coercions (which does not hold with respect to terms in general \cite{delataillade})  and this  will also be the basis for our interpretations.  We will use the equational formulation of dinaturality in System \fsub\!, introduced   by Cardelli et. al \cite{fsub},
to give related interpretations of $\forall^\k$ and $\forall^\top$  in terms of unbounded quantification and a meet operation ($\wedge$) for the subtyping relation. 

We now define the target calculus for this translation. Let System \fwedge (cf \cite{piercethesis}) be  System \fsub\!, restricted to unbounded quantification, and extended with a  binary meet operation on types  --- i.e. its raw types are given by the grammar:
$$T::= \top\ |\ X\ |\ T \rightarrow T\ |\ \forall X.T\ |\ T \wedge T$$
The type-formation rules of Table \ref{tfr} are extended with the rule:
\begin{center}
\AxiomC{  $\Theta \vdash S$}\AxiomC{ $\Theta \vdash T$}\BinaryInfC{ $ \Theta \vdash S \wedge T$}\DisplayProof
\end{center} 
and the subtyping rules of Table \ref{stj} with the rules:
\begin{center}
\AxiomC{ $\Theta \vdash S \wedge S'$} \UnaryInfC{ $ \Theta \vdash S \wedge S' \leqt S$}\DisplayProof \ \ \AxiomC{  $\Theta \vdash S \wedge S'$}\UnaryInfC{ $\Theta \vdash  S \wedge S' \leqt S'$}\DisplayProof\ \   \AxiomC{ $\Theta \vdash T \leqt S$}\AxiomC{ $\Theta \vdash T \leqt S'$}\BinaryInfC{ $ \Theta \vdash T \leqt S \wedge S'$}   \DisplayProof   
\end{center}
The grammar of raw terms remains unchanged (except that annotating types range over System \fwedge types) and the typing rules of Table \ref{tj} are also  unchanged. 

As observed in \cite{piercethesis}, a bounded variable  $\Theta, X \leqt S \vdash X$ may be represented by the type $\Theta, X \leqt \top \vdash X \wedge S$. This is bounded above by $S$ ($\Theta,X \leqt \top \vdash X \wedge S \leqt S$) and if $\Theta \vdash S' \leqt S$ then $\Theta \vdash S' \leqt S' \wedge S \leqt S'$, so that substituting  $S'$ for $X$ in $T[X]$ is equivalent to substituting $S'$ for $X$ in $T[X \wedge S]$, up to the equivalence on types induced by the subtype preorder. This suggests an  interpretation of the bounded quantification  $\Theta \vdash \forall (X \leqt S).T$ as $ \forall X.T[X \wedge S/X] $: 
\begin{lemma}\label{kfun} If $\Theta,X \leqt S \vdash T \leqt T'$ then $\Theta \vdash \forall X.T[X \wedge S/X] \leqt \forall X.T'[X \wedge S/X]$. 
\end{lemma}
\begin{proof}
We show that if  $\Theta,X \leqt S,\Theta' \vdash T \leqt T'$ then $\Theta,X \leqt \top,\Theta'[X\wedge S/X] \vdash T[X \wedge S/X] \leqt T'[X \wedge S/X]$ by induction on the derivation of $\Theta,X \leqt S,\Theta' \vdash T' \leqt T'$. 
\end{proof}
So if $\Theta,X \leqt S,\Theta' \vdash T \leqt T'$ then $\Theta \vdash \forall X.T[X \wedge S/X] \leqt  \forall X.T'[X \wedge S/X]$ --- i.e. this
interpretation satisfies the subtyping rule $\mathsf{\forall-Fun}$. 
Accordingly, we define:
\begin{definition}Let $\ulcorner \forall^\k (X \leqt S).T \urcorner \triangleq \forall X.T[X \wedge S/X]$.  
\end{definition}

\subsection{Interpretation of $\forall^\top$} 
The above interpretation of bounded quantification does not satisfy the $\mathsf{\forall - Top}$ subtyping rule because it is not antitone in the variable bound --- e.g. $S' \leqt S$ does not generally imply  $\ulcorner \forall^\k (X \leqt S).X \rightarrow X \urcorner = \forall X.X \wedge S \rightarrow X \wedge S \not \leqt  \forall X.X \wedge S' \rightarrow X \wedge S' =  \ulcorner \forall^\k (X \leqt S').X \rightarrow X \urcorner $.
The problem is that substitution into a type is not antitone with respect to the subtyping order: from a semantic viewpoint, types do not act as contravariant (nor covariant) functors with respect to the subtype preorder. A solution is to separate \emph{positive} and \emph{negative} occurrences of type-variables, such that substitution of the former is monotone, and of the latter is antitone with respect to the subtype preorder. In other words  types act as mixed-variance functors with respect to substitution of type variables (in the next section, we will see that this leads to an interpretation of terms as dinatural transformations).       Accordingly, we define a mixed substitution operation on types of System \fwedge which acts separately on the positive and negative occurences of type variables.
\begin{definition}Given raw types $S_-,S_+,T$, we assume, by $\alpha$-conversion, that neither $X$ nor any free variables of $S_,S_+$ are quantified in $T$.
 Let $T[(S_-,S_+)/X]$ be the substitution of $S_-$ and $S_+$ for the negative and  positive occurrences of $X$ in $T$, respectively, so that $T[S/X] = T[(S,S)/X]$.     
Formally: \begin{center}
\begin{tabular}{c c c}
$Y[(S_-,S_+)/X] =  \begin{cases} S_+ & \text{ if } Y \equiv X\\Y & \text{ otherwise} \end{cases}$ & & $\top[(S_-,S_+)/X] = \top$ \\
$(T \rightarrow T')[(S_-,S_+/X] = T[(S_+,S_-)/X] \rightarrow T'[(S_-,S_+)/X]$ & & $(\forall Y.T)[(S_-,S_+)/X] = \forall Y.T[(S_-,S_+)/X]$ \\
$(T \wedge T')[(S_-,S_+)/X] = T [(S_-,S_+)/X]\wedge T'[(S_-,S_+)/X]$    
\end{tabular}
\end{center}

\end{definition}  
This is antitone in $S_-$ and monotone in $S_+$ and $T$ with respect to the subtyping preorder. 
\begin{lemma}\label{fsubtop1}If $\Theta,\Theta'   \vdash S_-'\leqt S_-$ and  $\Theta,\Theta'   \vdash S_+\leqt S_+'$ then for any \fwedge-type $\Theta,X \leqt \top \vdash T$,  $$\Theta,\Theta' \vdash T[(S_-,S_+)/X] \leqt  T[(S_-',S_+')/X]$$ 
\end{lemma}
\begin{proof}
By induction on the size of $T$. 
\end{proof}
\begin{lemma}\label{fsubtop2}For \fwedge-types $\Theta, X\leqt \top \vdash T \leqt T'$ and any types $\Theta,\Theta' \vdash S_-, S_+$,   $$\Theta,\Theta' \vdash T[(S_-,S_+)/X] \leqt T'[(S_-,S_+)/X].$$
\end{lemma}
\begin{proof}
By induction on the derivation of $\Theta, X\leqt \top \vdash T_1 \leqt T_2$. 
\end{proof}

Thus we may interpret  $\forall^\top (X\leqt S).T$ in terms of unbounded quantification  by substituting only negative occurrences of $X$ in $T$ with $X \wedge S$. 
\begin{definition}Writing $T[S/X^-]$ for $T[(S,X)/X]$, let $\ulcorner \forall^\top (X\leqt S).T \urcorner \triangleq \forall X.T[X \wedge S/X^-]$.
\end{definition}
This interpretation of bounded quantification satisfies the $\mathsf{\forall-Top}$ subtyping  rule:
\begin{lemma}If $\Theta \vdash S' \leqt S$ and $\Theta, X\leqt \top \vdash T \leqt T'$ then $\Theta \vdash \forall X.T[X \wedge S/X^-]  \leqt  \forall X.T[X \wedge S'/X^-]$   
\end{lemma}
\begin{proof}
$\Theta \vdash S' \leqt S$ implies $\Theta,X \leqt \top \vdash  X \wedge S' \leqt X \wedge S$.\\
So $\Theta, X \leqt \top \vdash T[X \wedge S/X^-] \leqt T[X \wedge S'/X^-]$ by Lemma \ref{fsubtop1},\\
$\Theta, X \leqt \top \vdash T[X \wedge S'/X^-] \leqt T'[X \wedge S'/X^-]$ by Lemma \ref{fsubtop2}\\
 and so $\Theta \vdash \forall X.T[X \wedge S/X^-]  \leqt  \forall X.T'[X \wedge S'/X^-]$  as required. 
\end{proof}
This rule does not satisfy $\mathsf{\forall-Fun}$ because positive occurences of a variable are not interpreted as subtypes of their bounds.  For example, $\ulcorner \forall^\top (X \leqt S).X\urcorner $ is not a subtype of $\ulcorner \forall^\top (X \leqt S).S \urcorner$.
\subsection{Relating $\forall^\k$ and $\forall^\top$}
From the interpretations of $\forall^\k$ and $\forall^\top$ within the target calculus System \fwedge, we derive a subtyping rule which relates them.  
\begin{lemma}If  $\Theta \vdash T_0 \leqt S_0$ and $\Theta,  X\leqt S_0  \vdash S_1 \leqt T_1$ then $\Theta \vdash \forall X.  S_1[X \wedge S_0/X] \leqt \forall X.T_1[X \wedge T_0/X^-]$. 
\end{lemma}
\begin{proof}$\Theta, X \leqt \top \vdash  S_1[X \wedge S_0/X] \leqt  T_1[X \wedge S_0/X]$ by Lemma \ref{kfun}.\\
$\Theta,X \leqt \top \vdash X\wedge T_0 \leqt X \wedge S_0$ and $\Theta,X \leqt \top \vdash X \wedge S_0 \leqt X$ implies \\ 
$\Theta,X\leqt \top \vdash T_1[X \wedge S_0/X] = T_1[(X\wedge S_0,X\wedge S_0)/X]    \leqt T_1[(X \wedge T_0,X)/X] = T_1[X \wedge T_0/X^-]$ by Lemma \ref{fsubtop1}.\\
So $\Theta \vdash \forall X. S_1[X \wedge S_0/X] \leqt \forall X.T_1[X \wedge T_0/X^-]$ as required. 
\end{proof}
In other words, the semantics soundly interprets the following subtyping rule:
\begin{center}
\begin{prooftree}
 \AxiomC{$\Theta \vdash T_0 \leqt S_0$}\AxiomC{$  \Theta, X \leqt S_0  \vdash S_1 \leqt T_1$}\RightLabel{$\mathsf{\forall-Loc}$} \BinaryInfC{$ \Theta \vdash \forall^\k (X \leqt S_0). S_1 \leqt \forall^\top (X \leqt T_0). T_1$}
\end{prooftree}
\end{center}  
This is the rule $\mathsf{\forall-Loc}$ considered (without the  decorating superscripts)  as yet another candidate replacement for the original rule for subtyping bounded quantification in System \fsub  \cite{fsubtop}. However, the resulting type system lacks an evident subtyping algorithm, due to the failure of a key transitivity property which is essential to the subtyping algorithm for System \fsub\!. In Section \ref{subt} we show that the decorated form of the rule avoids this problem, allowing adaptation of the \fsub subtyping algorithm to the decorated calculus.
\section{System \fsubu}
We may now formally define a type system (System \fsubu\!) with   $\{\k,\top\}$-decorated bounded quantification. Raw types  are given by the grammar:
 $$T::= \top\ |\ X\ |\ T \rightarrow T\ |\ \forall^\k (X \leqt T).T\ |\ \forall^\top (X \leqt T).T $$
Subtyping judgments are given by the rules in Table \ref{stjd}, which replace the single original typing rule for bounded quantification of System \fsub with the rules  $\mathsf{\forall-Top}$, $\mathsf{\forall-Fun}$ and  $\mathsf{\forall-Loc}$.

\begin{table}
\begin{center}
 \AxiomC{\small $ \Theta, X \leqt T,\Theta' \vdash \top$}\RightLabel{\small $\mathsf{Var}$}  \UnaryInfC{\small $\Theta, X \leqt T,\Theta' \vdash X \leqt T $} \DisplayProof\ \  \AxiomC{\small $\Theta \vdash T$}\RightLabel{\small $\mathsf{Top}$} \UnaryInfC{\small $\Theta \vdash T \leqt \top$}\DisplayProof\ \   \AxiomC{\small $\Theta \vdash T$}\RightLabel{\small $\mathsf{Refl}$}  \UnaryInfC{\small $\Theta \vdash T \leqt T$}\DisplayProof\ \ 
  \AxiomC{\small $\Theta \vdash T \leqt T'$} \AxiomC{\small $  \Theta \vdash T' \leqt T''$} \RightLabel{\small $\mathsf{Trans}$} \BinaryInfC{\small$\Theta \vdash T \leqt T''$}\DisplayProof
\end{center}
\begin{center}
  \AxiomC{\small $\Theta \vdash S' \leqt S \ \  \Theta \vdash T \leqt T'$}\RightLabel{\small $\mathsf{\rightarrow}$}  \UnaryInfC{\small $ \Theta \vdash S \rightarrow T \leqt S' \rightarrow T'$}\DisplayProof \ \ 
\AxiomC{\small $ \Theta, X \leqt S \vdash T \leqt T'$} \RightLabel{\small $\mathsf{\forall-Fun}$} \UnaryInfC{\small $\Theta \vdash \forall^\k (X \leqt S).T \leqt \forall^\k(X \leqt S).T'$}
\DisplayProof
\end{center}
\begin{center}
\AxiomC{\small $\Theta \vdash T_0 \leqt S_0$}\AxiomC{\small $  \Theta, X \leqt S_0 \vdash S_1 \leqt T_1$}\RightLabel{\small $\mathsf{\forall-Loc}$} \BinaryInfC{\small $ \Theta \vdash \forall^\k (X \leqt S_0).S_1 \leqt \forall^\top(X \leqt T_0).T_1$}
\DisplayProof
\ \ 
\AxiomC{\small $\Theta \vdash T_0 \leqt S_0$}\AxiomC{\small $ \Theta, X \leqt \top \vdash S_1 \leqt T_1$} \RightLabel{\small $\mathsf{\forall-\top}$} \BinaryInfC{\small $\Theta \vdash \forall^\top (X \leqt S_0).S_1 \leqt \forall^\top(X \leqt T_0).T_1$}
\DisplayProof
\caption{Subtyping Rules for System \fsubu}   \label{stjd}
\end{center}
\end{table}
Raw terms are defined as in System \fsub (except that annotating types range over the raw \fsubu-types). Type-variable abstraction is typed using  $\forall^\k$  and instantiation is typed using  $\forall^\top$: 
  the typing rules (Table \ref{tjd}) are obtained  
 by decorating the introduction and elimination rules for bounded quantification with $\k$ and $\top$, respectively. Introduction of $\forall^\top$ and elimination of $\forall^\k$ are derivable by subsumption, e.g.
{ \begin{prooftree}
\AxiomC{$\Theta,X\leqt S \vdash t:T$}\RightLabel{$\mathsf{\forall-i}$}
\UnaryInfC{$\Theta \vdash \Lambda (X \leqt S).t:  \forall^\k(X \leqt S).T$}    
\AxiomC{}\RightLabel{$\mathsf{Refl}$} 
\UnaryInfC{$\Theta \vdash S \leqt S$} 
\AxiomC{} \RightLabel{$\mathsf{Refl}$}
\UnaryInfC{$\Theta,X \leqt S \vdash T \leqt T$}\RightLabel{$\mathsf{\forall-Loc}$} 
\BinaryInfC{$\Theta  \vdash \forall^\k (X \leqt S).T \leqt \forall^\top (X \leqt S).T$}\RightLabel{$\mathsf{sub}$}
\BinaryInfC{$\Theta \vdash \Lambda (X \leqt S).t:  \forall^\top(X \leqt S).T$} 
\end{prooftree}}


\begin{table}
\begin{center}
\AxiomC{\small$\Theta \vdash \top$}\RightLabel{\small $\top$}\UnaryInfC{\small $\Theta \vdash \topp:\top$}\DisplayProof\ \  
\AxiomC{\small $\Theta,x:T,\Theta' \vdash \top$} \RightLabel{\small $\mathsf{var}$}\UnaryInfC{\small $\Theta,x:T,\Theta' \vdash x:T$}\DisplayProof\ \   \AxiomC{\small $\Theta \vdash t:T$} \AxiomC{\small$\Theta \vdash T \leqt T'$}\RightLabel{$\mathsf{sub}$}\BinaryInfC{\small $\Theta  \vdash t:T'$}\DisplayProof
\end{center}
\begin{center}
\AxiomC{\small $\Theta,x:S \vdash  t:T$}\RightLabel{\small $\mathsf{\rightarrow - i}$} \UnaryInfC{\small $  \Theta\vdash \lambda x:S.t:S \rightarrow T$}\DisplayProof \ \  \AxiomC{\small $\Theta \vdash  t:S \rightarrow T$} \AxiomC{\small $\Theta \vdash s:S$} \RightLabel{\small $\mathsf{\rightarrow - i}$}  \BinaryInfC{\small $\Theta \vdash t\spc s:T$} \DisplayProof
\end{center}
\begin{center}
\AxiomC{\small $\Theta, X \leqt S \vdash t:T$}\RightLabel{\small $\mathsf{\forall - i}$} \UnaryInfC{\small$ \Theta \vdash \Lambda (X \leqt S).t:\forall^\k (X \leqt S).T$}\DisplayProof \ \  \AxiomC{\small $\Theta \vdash t:\forall^\top (X \leqt S).T$} \AxiomC{\small $  \Theta \vdash S' \leqt S $} \RightLabel{\small $\mathsf{\forall - i}$} \BinaryInfC{\small $\Theta \vdash t\{S'\}:T[S'/X]$}\DisplayProof
\end{center}
\caption{Typing Judgments for System \fsubu }\label{tjd}

\end{table}



The typing of Ghelli's example (Proposition \ref{ghelli}) in System \fsubu illustrates how it repairs the failure of the minimal typing property in \fsubtop\!. Recall that the term-in-context $X \leqt \top \vdash \Lambda (Z \leqt X).\lambda (y:Z).y$ may be typed with either of the \fsubtop-types $X \leqt \top \vdash \forall^\top (Z \leqt X). (Z \rightarrow  Z)$ and $X \leqt \top \vdash \forall^\top (Z \leqt X). (Z \rightarrow  X)$, which are not bounded below by any \fsubtop-type. 

In \fsubu\!, $X \leqt \top \vdash \Lambda (Z \leqt X).\lambda (y:Z).y$ has the minimal  type $X \leqt \top \vdash \forall^\k (Z \leqt X). (Z \rightarrow  Z)$: using the $\mathsf{\forall-Loc}$ rule, we may derive both:
{{\begin{center}
\begin{prooftree}
 \AxiomC{$X \leqt \top \vdash X \leqt X$}  \AxiomC{$X \leqt \top,Z \leqt X \vdash Z \rightarrow Z \leqt  Z \rightarrow Z$} \RightLabel{$\mathsf{\forall-Loc}$}
\BinaryInfC{$X \leqt \top \vdash \forall^\k (Z \leqt X).(Z \rightarrow Z)  \leqt \forall^\top (Z \leqt X).(Z \rightarrow Z) $}  
\end{prooftree}
\end{center}}}
and 
{{\begin{center}
\begin{prooftree}
 \AxiomC{$X \leqt \top \vdash X \leqt X$}  \AxiomC{$X \leqt \top,Z \leqt X \vdash Z \rightarrow Z \leqt  Z \rightarrow X$}  \RightLabel{$\mathsf{\forall-Loc}$}
\BinaryInfC{$X \leqt \top \vdash \forall^\k (Z \leqt X).(Z \rightarrow Z)  \leqt \forall^\top (Z \leqt X).(Z \rightarrow X) $}  
\end{prooftree}
\end{center}}}
We will show that  System \fsubu possesses the minimal typing property in Section \ref{subt}.

\section{Semantics of System \fsubu}     
We will consider subtyping and typechecking algorithms for System \fsubu in the following section, showing that these are quite well-behaved. First, we take its semantic justification further, by showing that the  interpretation of $\forall^\k$ and $\forall^\top$ in System \fwedge may be extended to terms.

We interpret the $\forall^\top$-elimination rule, using the fact that $\Theta \vdash S' \leqt S$ implies $\Theta \vdash S' \wedge S \leqt S' \leqt S'\wedge S$.
\begin{proposition}If $\Theta \vdash t: \ulcorner \forall^\top X \leqt S.T \urcorner$ and $\Theta \vdash S' \leqt S$  then $\Theta \vdash  t\{S'\}:T[S'/X]$. 
\end{proposition}
\begin{proof}
From  $\Theta \vdash S' \leqt S$ (and $\Theta \vdash S'\leqt S'$) we may infer
$\Theta \vdash S' \leqt S' \wedge S$ and hence by Lemma \ref{fsubtop1}  $\Theta \vdash T[(S' \wedge S, S')/X] \leqt T[S'/X] = T[(S',S')/X]$. Hence we have the following derivation of $\Theta \vdash  t\{S'\}:T[S'/X]$:
{ \begin{prooftree}
\AxiomC{$\Theta \vdash t: \forall X \leqt \top.T[X\wedge S/X^-]$}
\AxiomC{$\Theta \vdash S' \leqt \top$}\RightLabel{$\mathsf{\forall-e}$} 
\BinaryInfC{$\Theta \vdash t\{S'\}:T[(S'\wedge S,S')/X]$}    
\AxiomC{$\Theta \vdash T[(S' \wedge S, S')/X] \leqt T[S'/X]$}  \RightLabel{$\mathsf{sub}$} 
\BinaryInfC{$\Theta \vdash t\{S'\}: T[S'/X]$}
\end{prooftree}}
\end{proof}
To derive the $\forall^\k$-introduction rule, we extend the interpretation to bounded type-abstraction:
\begin{definition}Suppose $\Theta,X \leqt S \vdash t:T$. Let $\ulcorner \Lambda (X \leqt S).t \urcorner \triangleq \Lambda (X\leqt \top).t[X\wedge S/X]$.   
\end{definition}
A straightforward induction establishes that:  
\begin{lemma}\label{subst}$\Theta,  X\leqt S,\Theta'   \vdash t:T$ implies $\Theta, X\leqt \top, \Theta'[X \wedge S/X]   \vdash t[X \wedge S/X] : T[X \wedge S/X]$.  
\end{lemma}
and hence:
\begin{proposition}If $\Theta,  X\leqt S \vdash t:T$ then $\Theta  \vdash \ulcorner \Lambda (X \leqt S).t \urcorner  : \ulcorner \forall^\k X \leqt S.T \urcorner$.
\end{proposition}

\subsection{Soundness via Dinaturality}
We show that these interpretations of type-abstraction and instantiation are sound with respect to $\beta$ and $\eta$ equivalences using 
 the equational theory for System \fsub  introduced by Cardelli et. al. \cite{fsub}, which we adapt to System \fsubu\!.  Derivation rules (for equational judgements in context $\Theta \vdash t = t':T$, where $\Theta \vdash t:T$ and  $\Theta \vdash t':T$) are given in Table \ref{eqr}.   
\begin{table}
\begin{center}

\AxiomC{\small $\Theta \vdash t:\top$}\AxiomC{ \small$\Theta \vdash t':\top$}\RightLabel{\small $\mathit{\top}$} \BinaryInfC{\small $ \Theta \vdash t = t':\top$} \DisplayProof\ \ \  \AxiomC{\small $\Theta \vdash t:T$} \RightLabel{\small $\mathit{refl}$}\UnaryInfC{\small$\Theta \vdash t = t:T$}  \DisplayProof\ \ \ 
\AxiomC{\small$\Theta \vdash r = s:T$}\AxiomC{\small $  \Theta \vdash s = t:T$} \RightLabel{\small $\mathit{trans}$} \BinaryInfC{\small $ \Theta \vdash r = t:T$} \DisplayProof\ \ \ 
\AxiomC{\small $\Theta \vdash s = t:T$} \RightLabel{\small $\mathit{sym}$}\UnaryInfC{\small$\Theta \vdash t = s:T$} \DisplayProof
\end{center}
\begin{center}
\AxiomC{\small $\Theta,x:S \vdash t:T$}\AxiomC{\small$ \Theta \vdash s:S$}\RightLabel{\small $\mathit{\beta_1}$}\BinaryInfC{\small $ \Theta \vdash (\lambda x.t)\spc s = t[s/x]:T$}\DisplayProof \ \ \ \AxiomC{\small $\Theta \vdash t:S \rightarrow T$}  \LeftLabel{\small $x\not \in \dom(\Theta)$} \RightLabel{\small $\mathit{\eta_1}$} \UnaryInfC{ \small $\Theta \vdash \lambda x.(t\spc x) = t: S \rightarrow T$} \DisplayProof \end{center}
\begin{center}
\AxiomC{\small $\Theta,X \leqt S \vdash  t:T$}\AxiomC{\small $ \Theta \vdash R \leqt S$} \RightLabel{\small $\mathit{\beta_2}$} \BinaryInfC{ \small $\Theta \vdash (\Lambda (X \leqt S).t)\{R\} =   t[S'/X]:T[S'/X]$}\DisplayProof\ \ \  \AxiomC{\small $\Theta \vdash t:\forall^\top X \leqt S.T$} \LeftLabel{\small $Y\not \in \dom(\Theta)$}\RightLabel{\small $\mathit{\eta_2}$} 
\UnaryInfC{\small $ \Theta \vdash \Lambda (Y\leqt S).(t\{Y\}) = t: \forall^\top X \leqt S.T$} \DisplayProof
\end{center}
\begin{center} 
\AxiomC{\small $\Theta,X:S \vdash t =t':T$}\RightLabel{\small$\mathit{abs_1}$} \UnaryInfC{\small$ \Theta \vdash \lambda x:S.t = \lambda x.t':S \rightarrow T$}  \DisplayProof\ \ \   \AxiomC{\small $\Theta,X \leqt S \vdash t = t':T$}\RightLabel{\small $\mathit{abs_2}$} \UnaryInfC{\small $\Theta \vdash \Lambda (X \leqt S).t = \Lambda (X \leqt S).t':\forall^\k X \leqt S.T$}\DisplayProof\ \ \ 
\AxiomC{\small $\Theta \vdash t = t':S \rightarrow T $}\AxiomC{\small $\Theta \vdash s = s':S$}\RightLabel{\small $\mathit{app_1}$} \BinaryInfC{ \small$\Theta \vdash t\spc s = t'\spc s':T$}\DisplayProof  
\end{center}
\begin{center}
\AxiomC{\small $\Theta \vdash t = t':\forall^\top X \leqt S.T $}\AxiomC{\small$\Theta \vdash R \leqt S$}\AxiomC{\small$ \Theta \vdash R' \leqt S$}  \AxiomC{\small $ \Theta \vdash T[R/X] \leqt  T'$} \AxiomC{\small $\Theta \vdash T[R'/X]  \leqt T'$} \RightLabel{\small$\mathit{app_2}$}\QuinaryInfC{\small$\Theta \vdash t\{R\} = t'\{R'\}:T'  $}\DisplayProof\end{center}
\begin{center}
\caption{Derivation Rules for Term Equivalence}\label{eqr}
\end{center}
\end{table}
The rules for original System \fsub\!, and the target calculus System \fwedge, may be obtained by simply erasing the decorations on quantifiers. They axiomatize term-equality as a congruence containing  $\beta$ and $\eta$ equalities for type and term variable abstraction, together with the  rule $\mathit{app_2}$:
\begin{center}
\AxiomC{$\Theta \vdash t = t':\forall X \leqt S.T$}\AxiomC{$  \Theta \vdash R \leqt S$}\AxiomC{$ \Theta \vdash R' \leqt S$}\AxiomC{$\Theta \vdash T[R/X] \leqt  T'$}\AxiomC{$ \Theta \vdash T[R'/X]  \leqt T'$} \QuinaryInfC{$\Theta \vdash t\{R\} = t'\{R'\}:T'  $}\DisplayProof
\end{center}
from which reflexivity and $\beta$-equivalence yield the derived  rule:
\begin{center}
\AxiomC{ $\Theta,X \leqt S \vdash t:T$} \AxiomC{$ \Theta \vdash R \leqt S$}\AxiomC{$\Theta \vdash R'\leqt S$}\AxiomC{$\Theta \vdash T[R/X] \leqt T'$}\AxiomC{$   \Theta \vdash T[R'/X] \leqt T'$} \RightLabel{$(*)$} \QuinaryInfC{$ \Theta \vdash t[R/X] = t[R'/X]:T' $}\DisplayProof  
\end{center}
This relates subtyping to parametricity by expressing the \emph{extranaturality} of type instantiation with respect to subsumption.  In the setting of the $\lambda$-calculus, this is equivalent to \emph{dinaturality}: in category-theoretic terms, dinaturality  generalizes the notion of natural transformation between covariant functors to mixed-variance functors \cite{Bainbridge}. Since second-order types correspond to mixed-variance functors on the preorder of subtypes in a given context (Lemma \ref{fsubtop1}), dinaturality of instantiation 
 may be captured within the equational theory for System \fwedge\ \! by the derived rule: 

\begin{center}
{\LARGE $\Theta,X \leqt \top  \vdash t:S \rightarrow T \ \ \ \Theta \vdash R \leqt R'  \over \Theta \vdash t[R/X] = t[R'/X]:S[(R',R)/X] \rightarrow  T[(R,R')/X] $}   
\end{center}

Diagrammatically, this is the commuting hexagon:
$$\xymatrix@R=15pt@C=18pt{ &  S[R/X] \ar[r]^{t[R/x]} &  T[R/X] \ar[rd]^{\leqt} \\
S[(R',R)/X] \ar[ru]^{\leqt}\ar[rd]_{\leqt} & & & T[(R,R')/X]\\
&  S[R'/X] \ar[r]^{t[R'/x]} &  T[R'/X] \ar[ru]_{\leqt}}$$

The dinaturality expressed in equation $(*)$  is crucial to showing that interpretation in the target language System \fwedge \ \! is sound with respect to the equational theory, because it equates terms which are instantiated with types which are equivalent up to subtyping equivalence (i.e. $S \leqt T$ and $T \leqt S$).
 Consider, for example, second-order $\beta$-equality: 
\begin{proposition} If  $\Theta,X \leqt S \vdash t:T$ and $\Theta \vdash S' \leqt S$ then $\Theta \vdash \ulcorner \Lambda (X \leqt S).t\urcorner\{S'\} = t[S'/X]:T[S'/X]$.  
\end{proposition}
\begin{proof}
If $\Theta, X \leqt S \vdash  t:T$ then  $\Theta, X \leqt \top  \vdash  t:T[X \wedge S/X]$ by Lemma \ref{subst}, and the $\beta$-equivalence rule is:
{
\begin{prooftree}
\AxiomC{$\Theta,X \leqt \top  \vdash  t[X\wedge S/X]:T[S'/X]$}
\AxiomC{$\Theta \vdash S' \leqt \top$}\RightLabel{$\mathit{\beta_2}$} 
\BinaryInfC{$\Theta \vdash (\Lambda X.t[X \wedge S/X])\{S'\} = t[X\wedge S/X][S'/X]$}
\end{prooftree}}
\noindent where $t[X\wedge S/X][S'/X] \equiv t[S' \wedge S/X]$. So we need to show that $t[S'\wedge S/X]$ is equivalent to $t[S'/X]$ at type $T[S'/X]$. 

$\Theta \vdash S' \wedge S \leqt S'$, and $\Theta \vdash S' \leqt S $  implies   $\Theta \vdash S' \leqt S' \wedge S$, and so by Lemma \ref{fsubtop1} $\Theta \vdash T[S \wedge S'/X] \leqt T[S'/X]$. Thus we have the following instance of our derived rule $(*)$
\begin{center}
{\footnotesize \begin{prooftree}
\AxiomC{$\Theta, X\leqt S \vdash t:T$}
 \AxiomC{$\Theta \vdash S'\wedge S \leqt S$} 
\AxiomC{$\Theta \vdash S' \leqt S$}
\AxiomC{$\Theta \vdash T[S \wedge S'/X] \leqt T[S'/X]$} 
\AxiomC{$\Theta \vdash T[S'/X] \leqt T[S'/X]$}    
\QuinaryInfC{$\Theta \vdash t[S \wedge S'] = t[S'/X]: T[S'/X]$} 
\end{prooftree}}
\end{center}
and so by transitivity:
{\footnotesize \begin{prooftree}
\AxiomC{$\Theta \vdash (\Lambda X.t[X \wedge S/X])\{S'\} =  t[S'\wedge S/X]$}
\AxiomC{$\Theta \vdash t[S \wedge S'] = t[S'/X]: T[S'/X]$} \RightLabel{$\mathit{trans}$} 
\BinaryInfC{$\Theta \vdash  (\Lambda X.t[X\wedge S/X])\{S'\} = t[S'/X]:T[S'/X]$}
\end{prooftree}}
\end{proof}
Similarly, we use $(*)$ to show soundness with respect to second-order $\eta$-equality (at $\forall^\top$-types):
\begin{proposition}If $\Theta \vdash t:\ulcorner \forall^\top (X \leqt S).T \urcorner$ and $Y \not \in \dom(\Theta)$, then $\Theta \vdash t = \ulcorner \Lambda (Y \leqt S).t\{Y\}\urcorner:\ulcorner \forall^\top (X \leqt S).T \urcorner$.    
\end{proposition}
\begin{proof}
The second-order $\eta$-equality rule itself yields:
{\footnotesize \begin{prooftree}
\AxiomC{$\Theta \vdash t:\forall X.T [X\wedge S/X^-]$}\RightLabel{$\mathit{\eta_2}$} 
\UnaryInfC{$\Theta \vdash t = \Lambda Y.t\{Y\}:\forall X.T[X \wedge S/X^-]$}
\end{prooftree}}
However, $\ulcorner \Lambda (Y \leqt S).t\{Y\}\urcorner = \Lambda Y.t\{S \wedge Y\}$. 
Noting that $\Theta, Y\leqt S \vdash X \wedge S \leqt (X \wedge S) \wedge S$ and hence by Lemma \ref{fsubtop1}, $\Theta, Y \leqt \top \vdash  T[((X \wedge S) \wedge S, X \wedge S)/X] \leqt T[(Y \wedge S, Y)/X]$, we may use dinaturality to infer:
\begin{center}{\footnotesize
\begin{prooftree}
\AxiomC{$\Theta \vdash t:\forall X.T [X\wedge S/X^-]$}
\UnaryInfC{$\Theta, Y \leqt \top \vdash t:\forall X.T [X\wedge S/X^-]$}

\AxiomC{$\Theta,Y \leqt \top \vdash T[(Y\wedge S,Y)/X],T[((Y \wedge S) \wedge S, Y \wedge S)/X]  \leqt   T[(Y \wedge S,Y)/X]$}
\RightLabel{$\mathit{app_2}$}
   \BinaryInfC{$\Theta, Y \leqt \top \vdash t\{Y\} = t\{Y \wedge S\}:T[(Y \wedge S,Y)/X]$}  \RightLabel{$\mathit{abs_2}$}   
\UnaryInfC{$\Theta \vdash \Lambda Y.t\{Y\} = \Lambda Y.t\{Y\wedge S\}:\forall X.T[X \wedge S/X^-]$}       

\end{prooftree}}
\end{center}
 and hence by transitivity:  
\begin{center}
{\footnotesize
\begin{prooftree}
\AxiomC{$\Theta \vdash t = \Lambda Y.t\{Y\}:\forall X.T[X \wedge S/X^-]$}
\AxiomC{$\Theta \vdash \Lambda Y.t\{Y\} = \Lambda Y.t\{Y\wedge S\}:\forall X.T[X \wedge S/X^-]$}\RightLabel{$\mathit{trans}$} 
\BinaryInfC{$\Theta \vdash t = \Lambda Y.t\{Y \wedge S\}:\forall X.T[X \wedge S/X^-] $}
\end{prooftree}}

\end{center}
\end{proof}

Soundness with respect to the remaining rules (including $\mathit{app_2}$ itself) is straightforward.
\section{Subtyping and Typechecking Algorithms for System \fsubu}
\label{subt}
Having established a semantic basis for System \fsubu\!\!, we now describe procedures for solving its subtyping and typechecking problems. These adapt readily from their analogues for System \fsub \cite{CGhe}: we give sufficient details here to show how the difficulties previously  associated with the rules   $\mathsf{\forall-Loc}$ and $\mathsf{\forall-Top}$ are avoided.
\begin{table}
\begin{center}
  \AxiomC{\small $\Theta \vdash T$} \UnaryInfC{\small $\Theta \vdash_A T \leqt \top$}\DisplayProof\ \   \AxiomC{\small $\Theta \vdash X$} \UnaryInfC{\small $\Theta \vdash_A X \leqt X$}\DisplayProof \ \ \AxiomC{\small $ \Theta, X \leqt S,\Theta' \vdash_A S \leqt T$}\RightLabel{\small $T \not \equiv \top,X$}  \UnaryInfC{\small $\Theta, X \leqt T,\Theta' \vdash_A X \leqt T $} \DisplayProof
\end{center}
\begin{center}\AxiomC{\small $\Theta \vdash_A S' \leqt S \ \  \Theta \vdash_A T \leqt T'$} \UnaryInfC{\small $ \Theta \vdash_A S \rightarrow T \leqt S' \rightarrow T'$}\DisplayProof \ \ 
\AxiomC{\small $ \Theta, X \leqt S \vdash_A T \leqt T'$}  \UnaryInfC{\small $\Theta \vdash_A \forall^\k (X \leqt S).T \leqt \forall^\k(X \leqt S).T'$}
\DisplayProof
\end{center}

\begin{center}

\AxiomC{\small $\Theta \vdash_A T_0 \leqt S_0$}\AxiomC{\small $  \Theta, X \leqt S_0 \vdash_A S_1 \leqt T_1$} \BinaryInfC{\small $ \Theta \vdash_A \forall^\k (X \leqt S_0).S_1 \leqt \forall^\top(X \leqt T_0).T_1)$}
\DisplayProof
\ \ 
\AxiomC{\small $\Theta \vdash_A T_0 \leqt S_0$}\AxiomC{\small $ \Theta, X \leqt \top \vdash_A S_1 \leqt T_1$} \BinaryInfC{\small $\Theta \vdash_A \forall^\top (X \leqt S_0).S_1 \leqt \forall^\top(X \leqt T_0).T_1)$}
\DisplayProof
\end{center}
\begin{center}
\caption{Algorithmic Derivation Rules for Subtyping Judgements}\label{suba}
\end{center}
\end{table}

The subtyping algorithm is given by defining an alternative ``algorithmic'' presentation of the subtyping relation, via a set of derivation rules (Table \ref{suba}) with the property that any subtyping judgment is the consequence of at most one rule, so that the evident search procedure for the derivation of a subtyping judgment in this system is deterministic. To show that this is sound and complete, it suffices to establish that the subtyping judgments derivable according to the rules in Tables \ref{stj} and \ref{suba} are the same, by showing that  each rule of one system is admissible in the other, and vice-versa. The only case that differs from the proof for System \fsub is to show that the transitivity rule for the subtyping relation ($\mathsf{Trans}$) is admissible in the algorithmic system.
\begin{lemma}\label{trans}If $\Theta \vdash_A R \leqt S$ and $\Theta \vdash_A S \leqt T$ then  $\Theta \vdash_A R \leqt T$.
\end{lemma} 
\begin{proof}By induction on the size of the derivations of $\Theta \vdash_A R \leqt S$ and $\Theta \vdash_A S \leqt T$. It follows the proofs for  Kernel \fsub or \fsubtop, except where $R$, $S$ and $T$ are differently quantified types (i.e. they are not all prefixed with quantifiers with the same decoration). Since it is never possible to infer $\forall^\top (X \leqt S_0). S_1 \leqt \forall^\k (X \leqt T_0). T_1$ (as this is not the consequence of any rule)   the only possibilities are:
\begin{itemize}
\item $R \equiv \forall^\k (X \leqt R_0).R_1$,  $S \equiv \forall^\k (X \leqt S_0).S_1$ and $T \equiv  \forall^\top (X \leqt T_0).T_1$. Then $R_0 \equiv S_0$ and  $\Theta \vdash_A T_0 \leqt S_0$ (so $\Theta \vdash_A T_0 \leqt R_0$)
 and $\Theta, X \leqt S_0 \vdash_A R_1 \leqt S_1$ and $\Theta, X \leqt S_0 \vdash_A S_1 \leqt T_1$. 
By induction hypothesis, $\Theta,X \leqt S_0 \vdash_A R_1 \leqt T_1$ and thus $\Theta \vdash_A \forall^\k (X \leqt R_0). R_1 \leqt  \forall^\top (X \leqt T_0).T_1$ as required.         
\item $R \equiv \forall^\k (X \leqt R_0).R_1$, and  $S \equiv \forall^\top (X \leqt S_0).S_1$ and $T \equiv  \forall^\top (X \leqt T_0).T_1$. Then $\Theta \vdash_A S_0 \leqt R_0$ and 
 $\Theta, X \leqt R_0 \vdash_A R_1 \leqt S_1$, and   $\Theta \vdash_A T_0 \leqt S_0$, and  $\Theta, X \leqt \top \vdash_A S_1 \leqt T_1$ (and hence  $\Theta, X \leqt R_0 \vdash_A S_1 \leqt T_1$). \\
 By induction hypothesis, $\Theta \vdash_A T_0 \leqt R_0$, and $\Theta,X \leqt R_0 \vdash_A R_1 \leqt T_1$ and thus $\Theta \vdash_A \forall^\k (X \leqt R_0). R_1 \leqt  \forall^\top (X \leqt T_0). T_1$ as required. 
\end{itemize}
\end{proof}
Admissibility of the transitivity rule is the property which fails for the restriction of (undecorated) System \fsub to $\mathsf{\forall-Loc}$:  note that the proof of Lemma \ref{trans} depends on the fact that if $\Theta \vdash_A R \leqt S$ and $\Theta \vdash_A S \leqt T$ are derivable, at most one of these derivations may terminate with the rule $\mathsf{\forall-Loc}$. Using \ref{trans}  we establish:
\begin{proposition}$\Theta \vdash S \leqt T$ if and only if $\Theta \vdash_A S \leqt T$.
\end{proposition}

\subsection{Typechecking}
\label{minl}
The algorithm for type-synthesis in System \fsubu is also an adaptation from the algorithm for System \fsub\!. It follows a similar pattern to the subtyping algorithm: Table \ref{mintype} gives rules for deriving a unique minimal type for each typable term  $\Theta \vdash t$, via judgments of the form $\Theta \vdash_M t:T$. These make use of the following operation. 
 \begin{lemma}\label{nonatom}For any type $\Theta \vdash T$, there exists a minimal non-atomic type  $\Theta \vdash \Theta^*(T)$ such that  $\Theta \vdash T \leqt \Theta^*(T)$.
\end{lemma}
\begin{proof}
Define $\Theta \vdash \Theta^*(T)$ by:
\begin{center}
$\Theta^*(T) = \begin{cases} \Theta^*(S) & \text{if } T \equiv X \text{ and } \Theta \equiv \Theta',X \leqt S, \Theta'' \\ T & \text{otherwise} \end{cases}$
\end{center}    
If $T$ is a  non-atomic type then it is immediate that $\Theta^*(T)$ is a $\leqt$-minimal type  such that  $\Theta \vdash T \leqt \Theta^*(T)$ (i.e. for any non-atomic type $T'$, if  $\Theta \vdash T \leqt T'$ then $\Theta \vdash \Theta^*(T) \leqt T'$).
 Otherwise, $T$ is a type-variable $X$, where $\Theta \equiv \Theta',X \leqt S, \Theta''$ and we prove the lemma by induction on the length of $\Theta'$. 
\end{proof}
Any  term-in-context may be pattern-matched to the conclusion of exactly one derivation rule in Table \ref{mintype}, yielding a deterministic algorithm for synthesizing a minimal type for each typable term-in-context. Soundness and completeness of this algorithm is established by showing that:
\begin{proposition}\label{minty}$\Theta \vdash t:T$ if and only if $\Theta \vdash_M t:S$ for some $S$ such that $\Theta \vdash S \leqt T$.
\end{proposition}
\begin{proof}
From right to left, it suffices to check that $\Theta \vdash t:T$ implies $\Theta \vdash t:S$ by observing that each rule of $\vdash_M$ is derivable in \fsubu\!, so that $\Theta \vdash S \leqt T$ implies $\Theta \vdash t:T$.
 
The proof of the implication from left to right is by induction on the length of derivation of $\Theta  \vdash t:T$. We consider the cases where the last rule in this derivation is introduction of  $\forall^\k$ or elimination of  $\forall^\top$. 

Suppose the last rule applied is introduction of $\forall^\k$\!, so that $t \equiv \Lambda (X \leqt S).t'$ and $T \equiv \forall^\k (X \leqt S).T'$, where $\Theta, X \leqt S \vdash t':T'$. By hypothesis, $\Theta,X \leqt S \vdash_M t':R$ for some $R$ such that   $\Theta, X \leqt S \vdash R \leqt T'$. Then  $\Theta \vdash_M t:\forall^\k(X \leqt S).R$ and $\Theta \vdash \forall^\k (X \leqt S).R \leqt T$ as required.       

Suppose the last rule applied is the elimination of $\forall^\top$\!, so that $t \equiv t'\{S'\}$ and  $T \equiv T'[S'/X]$, where $\Theta \vdash t: \forall^\top X \leqt S.T'$  and $\Theta \vdash S' \leqt S$. By inductive hypothesis, $\Theta \vdash_M t':R$ for some $R$ such that $\Theta \vdash R \leqt \forall^\top (X\leqt S).T'$. Then by Lemma \ref{nonatom}, $\Theta \vdash \Theta^*(R) \leqt  \forall^\top (X\leqt S).T'$, and hence $\Theta^*(R) \equiv \forall^\k (X \leqt S'').T''$ or $\Theta^*(R) \equiv \forall^\top (X \leqt S'').T''$  for some $S'',T''$ such that $\Theta \vdash S \leqt S''$ and  $\Theta, X\leqt S'' \vdash T' \leqt T''$.
Then $\Theta \vdash S' \leqt S''$ and  $\Theta \vdash_M t\{S'\}:T''[S'/X]$, where $\Theta \vdash T''[S'/X] \leqt T= T'[S'/X]$ as required.         
\end{proof}
Hence we have a sound and complete typechecking procedure for System \fsubu\!, as for System \fsub --- accept the typing $\Theta \vdash t:T$ if the minimal typing algorithm produces a typing $\Theta \vdash_M t:S$ and the subtyping algorithm accepts $\Theta \vdash_A S \leqt T$.
\begin{table}
\begin{center}

 \AxiomC{}\UnaryInfC{\small$  \Theta,x:T,\Theta' \vdash_M x:T$}\DisplayProof\ \   \AxiomC{\small $\Theta,x:S \vdash_M t:T$}\UnaryInfC{\small $ \Theta \vdash_M \lambda (x:S).t:S \rightarrow T$} \DisplayProof\ \  \AxiomC{\small $\Theta \vdash_M r:R$}\AxiomC{\small  $\Theta \vdash_M s:S$} \AxiomC{$\small \Theta \vdash S \leqt S'$} \RightLabel{\small $\Theta^*(R) = S' \rightarrow T$} \TrinaryInfC{\small$\Theta \vdash_M r\spc s:T$} \DisplayProof \end{center}
\begin{center}
 \AxiomC{}\UnaryInfC{\small $ \Theta \vdash_M \topp:\top $} \DisplayProof\ \  \AxiomC{\small $\Theta, X \leqt S \vdash_M t:T$} \UnaryInfC{\small $\Theta \vdash_M \Lambda (X \leqt S).t:\forall^\k(X \leqt S). T$} \DisplayProof\ \ 
    \AxiomC{\small$\Theta \vdash_M r:R$}\AxiomC{\small$ \Theta \vdash S \leqt S'$} \RightLabel{\small $\Theta^*(R) = \forall^\k (X \leqt S').T$ or $\Theta^*(R) = \forall^\top (X \leqt S').T$ }\BinaryInfC{\small$ \Theta \vdash_M r\{S\}:T[S/X]$}     \DisplayProof   

\end{center}
\begin{center}
\caption{Derivation Rules for Minimal Typing Judgments}\label{mintype}
\end{center}
\end{table}

\section{Typing System \fsubtop terms in System \fsubu}
Given a type of System \fsub\!, how should we decorate its  bounded quantifiers?  The obvious answer is  to do so uniformly --- i.e. choose either $\forall^\k$ or $\forall^\top$ for  all of the quantifiers in the types of variables, giving two possible translations of undecorated types (and type-annotated terms) into System \fsubu\!\!. 

The first of these choices  leads back  to Kernel \fsub\!. Writing $\Theta \vdash^\k T$ and $\Theta \vdash^\k t$ for Kernel \fsub types and terms  (i.e. all of the quantifiers in $\Theta$ and $T$ and $t$ are instances of $\forall^\k$), and 
$\Theta \vdash^\k S \leqt T$ and $\Theta \vdash^\k t:T$ for Kernel \fsub subtyping and typing judgments  (i.e. those derivable in Kernel \fsub\!, if the decorating superscripts are erased)
 it is straightforward to show that: 
\begin{itemize}
\item If
 $\Theta \vdash S \leqt T$, where  $\Theta \vdash^\k S,T$,  then $\Theta \vdash^\k S \leqt T$.
\item If $\Theta \vdash_M  t:S$, where  $\Theta \vdash^\k t$,  then  $\Theta \vdash^\k t:S$. 
\end{itemize}
and hence that  for Kernel \fsub types and terms, $\Theta \vdash t:T$ if and only if  $\Theta \vdash^\k t:T$. In other words, System \fsubu is a conservative extension of Kernel \fsub\!.

What of System \fsubtop?  It is again straightforward to show that System \fsubu conservatively extends the \fsubtop subtyping relation: for \fsubtop-types, $\Theta \vdash S \leqt T$ implies $\Theta \vdash^\top S \leqt T$.\footnote{Writing $\Theta \vdash^\top T$ and $\Theta \vdash^\top t$ if $\Theta \vdash T$ and $\Theta \vdash t$ are \fsubtop types and terms, and $\Theta \vdash^\top S \leqt T$ and $\Theta \vdash^\top t:T$ if these are \fsubtop subtyping and typing judgments (derivable in System \fsubtop if decorating superscripts are erased).} However,  there are \fsubtop-terms  which may be typed with a \fsubtop-type in System \fsubu but are  not typable in System \fsubtop itself. (Compare the following with the example given in  Proposition \ref{ghelli}.) 
\begin{proposition}System \fsubu is not a  conservative extension of System \fsubtop\!. 
\end{proposition}
\begin{proof}
Let $X \leqt \top \vdash u$ be the  \fsubtop-term $X \leqt \top \vdash^\top (\Lambda Y.\Lambda (Z \leqt X).\lambda (y:Y).y)\{X\}$. Then $X \leqt \top \vdash u:\forall^\top (Z \leqt X).Z \rightarrow X$, since we may derive the minimal type  $X \leqt \top \vdash_M u: \forall^\k (Z \leqt X).X \rightarrow X$ in \fsubu\!, and $X \leqt \top \vdash  \forall^\k (Z \leqt X).X \rightarrow X \leqt \forall^\top (Z \leqt X).Z \rightarrow X$ by $\mathsf{\forall-Loc}$,   
  so $X \leqt \top \vdash u:\forall^\top (Z \leqt X).Z \rightarrow X$ by subsumption. 

However, this typing is not valid in System \fsubtop\!. Suppose $X \leqt \top  \vdash^\top u:\forall^\top (Z \leqt X).Z \rightarrow X$. Then $X \leqt \top \vdash^\top \Lambda Y.\Lambda (Z \leqt X) \lambda(y:Y).y : \forall^\top Y.T$ for some \fsubtop-type $X \leqt \top,Y\leqt \top \vdash T$ such that $X \leqt \top \vdash T[X/Y]  \leqt \forall^\top (Z \leqt X).Z \rightarrow X$ --- i.e. $T[X/Y] \equiv \forall^\top(Z \leqt T_0).T_1 \rightarrow T_2$, where in particular $X \leqt \top, Y \leqt \top, Z \leqt \top \vdash Z \leqt T_1$, and so $T_1 \equiv Z$ or $T_1 \equiv \top$. But the only types for $\Lambda Y.\Lambda (Z \leqt X) \lambda(y:Y).y$ in System \fsubtop are $X \leqt \top \vdash \forall^\top Y.\forall^\top(Z \leqt X). Y \rightarrow Y$, $X \leqt \top \vdash \forall^\top Y.\forall^\top(Z \leqt X). Y \rightarrow \top$, $X \leqt \top \vdash \forall^\top Y.\top$ and $X \leqt \top \vdash \top$.\footnote{Since this term is $\beta$-normal,  by Proposition \ref{conservative} below we may use the type-synthesis algorithm for System \fsubu to derive these types.}   
\end{proof}
Note that  the $\beta$-normal form of $u$ ---  $X \leqt \top \vdash^\top \Lambda (Z \leqt X).\lambda (y:X).y$ --- \emph{is} typable in System \fsubtop with $X \leqt \top \vdash^\top \forall^\top (Z \leqt X).Z \rightarrow X$. In fact this is true in general: System \fsubu typing is conservative over System \fsubtop when restricted to $\beta$-normal forms, as we now show.
\begin{lemma}If  $\Theta \vdash_M t:T$, where $\Theta \vdash^\top t$ is  $\beta$-normal and not a ($\lambda$ or $\Lambda$)  abstraction, then  $\Theta \vdash^\top T$.\label{abs}
\end{lemma}
\begin{proof}By induction on the length of $t$.
\begin{itemize}
\item If $t \equiv \topp$, then $T \equiv \top$, which is a \fsubtop-type.
\item If $t \equiv x$ for some variable $x$, then $\Theta \equiv \Theta',x:T,\Theta''$ and so $T$ is a \fsubtop type by assumption.  
\item If $t \equiv t'\spc t''$ then $\Theta \vdash_M t':S \rightarrow T$ for some $S$.  Since $t$ is $\beta$-normal, $t'$ is not an abstraction. By hypothesis, $S \rightarrow T$ (and hence also $T$) is a \fsubtop type. 
\item If $t \equiv t'\{S\}$ then  $\Theta \vdash_M t': \forall^\top(X \leqt S').T'$, where $T'$ is an \fsubtop type (since $t$ is $\beta$-normal and not an abstraction). So $T \equiv T'[S/X]$ is a \fsubtop-type.   
\end{itemize}
\end{proof}
\begin{proposition}\label{conservative}
If $\Theta \vdash t:T$, where  $\Theta \vdash^\top t$ is $\beta$-normal and $\Theta \vdash^\top T$  then $\Theta \vdash^\top t:T$.     
\end{proposition}
\begin{proof}
By induction on the length of $t$. If  $T = \top$ then evidently    $\Theta\vdash^\top t:T$. Otherwise, suppose $\Theta \vdash_M t:S$ where  $\Theta \vdash S \leqt T$ and $T \not \equiv \top$:
\begin{itemize}
\item If $t \equiv x$ then $\Theta \equiv \Theta',x:S,\Theta''$ and so  $\Theta \vdash^\top t:S$ and  $\Theta \vdash^\top t:T$ by conservativity of subtyping. 
\item If $t \equiv t' \spc t''$ then $\Theta \vdash_M t':S'$, where $\Theta^*(S') = R \rightarrow S$ and $\Theta \vdash_M t'':R$ for some types $R,S'$. Then 
$t'$ is not an abstraction (as $t$ is $\beta$-normal) and so by Lemma \ref{abs},  $S'$ is a \fsubtop type, and hence so are $R \rightarrow S$ and $R$. By hypothesis  $\Theta \vdash^\top t':R \rightarrow S$. and 
$\Theta \vdash^\top t'':R$, so $\Theta \vdash^\top t:T$ as required. 
\item If $t \equiv \lambda (x:R).t'$ then  $S \equiv R \rightarrow S'$, where  
$\Theta,x:R \vdash_M t':S'$, and $T \equiv R' \rightarrow T'$, where $\Theta \vdash R \leqt R'$ and $\Theta \vdash S' \leqt T'$. By hypothesis $\Theta,x:R \vdash^\top t':S'$ and hence $\Theta, \vdash^\top t:T$ as required.

\item If $t \equiv t'\{R\}$ then since $t'$ is not an abstraction its minimal type is a \fsubtop type  by Lemma \ref{abs}. So $\Theta \vdash_M t':S'$, where $\Theta^*(S') = \forall^\top(X \leqt R').S''$ and $S \equiv S''[R'/X]$. By  hypothesis (since $S'$ is a \fsubtop type), $\Theta \vdash^\top t':S'$, and so $\Theta \vdash^\top t:T$ as required. 

\item If $t \equiv \Lambda (X \leqt R).t'$ then  $S \equiv \forall^\k (X \leqt R). S'$ and $T \equiv \forall^\top (X \leqt R'). T'$, for some $R,R',S',T'$  such that $\Theta, X \leqt S' \vdash_M t':R'$, 
$\Theta \vdash R\leqt R'$ and $\Theta,X \leqt R' \vdash S'\leqt T'$. Then by
 hypothesis $\Theta,X \leqt R' \vdash^\top t':T'$ and hence $\Theta \vdash^\top t:T$ as required.

\end{itemize}
\end{proof}
In semantic terms, System \fsubu is thus a conservative extension of System \fsubtop\!. However,  by supplying the missing minimal types it satisfies more cases of \emph{subject expansion}.

\section{Decidability of typechecking terms of \fsubtop in \fsubu}
 It is straightforward to show that (as in System \fsub\!) the typechecking algorithm for System \fsubu terminates on a given input if and only if every call made to the subtyping algorithm terminates. So  decidability of typechecking boils down to termination of these calls. We do not know whether the subtyping algorithm determined by the rules in Table \ref{suba} terminates in general, nor whether  a terminating algorithm exists. The culprit is the rule $\mathsf{\forall-Loc}$ used to infer the subtyping relation between types quantified by $\forall^\k$ and $\forall^\top$\!, which introduces a convoluted form of the rebounding problem encountered in System \fsub itself: $\mathsf{\forall-Loc}$ does not reduce the simple metrics on subtyping judgments used to prove termination for System \fsubtop or Kernel \fsub\!, but the arguments used to show undecidability of subtyping in System \fsubu do not apply \cite{fsubtop} either.

However, for uniformly decorated types this problem does not arise.  Decidability of  typechecking for Kernel \fsub terms follows by conservativity; here we show that typechecking of  \fsubtop-terms with \fsubtop-types is also decidable, by showing that the proof of decidability of subtyping for System \fsubtop shown in \cite{fsubtop} extends to the minimal types inferred for \fsubtop-terms  in System \fsubu\!.

\begin{definition}The \emph{minimal types for \fsubtop} are the \fsubu types  given by the grammar:
$$T::= S \ |\ \forall^\k(X \leqt S).T\ | \ S \rightarrow T$$
where $S$ ranges over the \fsubtop-types.
 \end{definition}
The following lemma justifies the terminology.
\begin{lemma}\label{fmtype}For any \fsubtop-term  $\Theta \vdash^\top t$, if  $\Theta \vdash_M t:T$ then $T$ is  a minimal type for \fsubtop\!.
\end{lemma}
\begin{proof}By induction on the length of $t$:
\begin{itemize}
\item If $t$ is a variable then its minimal type is that assigned to it in $\Theta$, which is a  \fsubtop-type.
\item If $t \equiv \Lambda (x:S).t'$ then $\Theta, X:S \vdash_M t':T'$, where $\Theta \vdash S$ is a \fsubtop-type and $\Theta \vdash T'$ is a minimal type for \fsubtop  by hypothesis, and so $T \equiv \forall^\k (X \leqt S).T'$ is a  minimal type for \fsubtop\!.   
\item If $t \equiv t'\{S\}$ then $\Theta \vdash_M t':\forall^\k(X \leqt S').T'$. By hypothesis  $\forall^\k(X \leqt S').T'$ is a minimal type for \fsubtop and hence so is $T'$. $S$ is a \fsubtop-type, and it is straightforward to check that $T \equiv T'[S/X]$ is therefore a minimal type for \fsubtop\!.
  \item The cases $t \equiv \top$, $t \equiv \lambda (x:S).t'$ and  $t \equiv t' \spc t''$  are  similar.
\end{itemize}
\end{proof}
\begin{proposition}\label{tsub}If $\Theta \vdash S$ is a minimal type for \fsubtop\!, and $\Theta \vdash^\top T$ is a \fsubtop-type, then the subtyping algorithm terminates on $\Theta \vdash_A S \leqt T$.   
\end{proposition}
\begin{proof}By induction on the size of $S$. If it is a \fsubtop-type then the subtyping algorithm terminates by the proof of Castagna and Pierce \cite{fsubtop}. The remaining cases are:
\begin{itemize}
\item $S \equiv \forall^\k (X \leqt S_0).S_1$. If $T \equiv \forall^\top (X \leqt T_0).T_1$ then by induction hypothesis the algorithm terminates on $\Theta \vdash_A T_0 \leqt S_0$ and $\Theta,X\leqt S_0 \vdash_A S_1 \leqt T_1$ and hence terminates on $\Theta \vdash_A S \leqt T$. Otherwise, either $T \equiv \top$, and so $\Theta \vdash_A S \leqt T$ is accepted, or $T \not \equiv \top$ and it is rejected immediately.
\item $S \equiv S_0 \rightarrow S_1$ (similar).       
\end{itemize}
\end{proof}

\begin{proposition}\label{mtt}The minimal typing  algorithm terminates on any  \fsubtop-term  $\Theta \vdash^\top t$.
\end{proposition}
\begin{proof}By induction on the size of $t$,  verifying that the minimal typing algorithm calls the subtyping algorithm only on terminating inputs. 

Suppose, for example, that $t \equiv t'\{S\}$.
  By induction, the algorithm either rejects $\Theta \vdash t'$  or finds a minimal typing $\Theta \vdash_M t:T'$, where
 $ \Theta \vdash T'$ is a  \fsubtop-minimal type by Lemma \ref{fmtype}, and hence so is  $\Theta^*(T')$. If $\Theta^*(T') \equiv \forall^\k (X \leqt T_0).T_1$ or $\Theta^*(T') \equiv \forall^\top (X \leqt T_0).T_1$   then  $T_0$ is a \fsubtop-type, and the subtyping algorithm either accepts or rejects $\Theta \vdash_A S \leqt T_0$ by Lemma \ref{tsub}: in the former case  the minimal typing algorithm  returns $T_1[S/X]$ as the minimal type of $t$, otherwise (or if  $\Theta^*(T')$ is not a bounded quantification) it rejects. 
\end{proof}

\begin{proposition}Typechecking of \fsubtop-terms in System \fsubu is decidable. 
\end{proposition}
\begin{proof}
For any \fsubtop-term $\Theta \vdash^\top t$ and \fsubtop-type $\Theta \vdash^\top T$,  by Proposition \ref{mtt} the minimal typing algorithm either rejects $\Theta \vdash t$ or  produces a \fsubtop-minimal type $\Theta \vdash_M:S$, in which case the the subtyping algorithm  either rejects  $\Theta \vdash_A S \leqt T$ or accepts  it --- and thus the typing $\Theta \vdash t:T$ --- by Lemma \ref{tsub}.
\end{proof}
By Proposition \ref{conservative} this extends to typechecking of $\beta$-normal terms in \fsubtop itself.
\begin{corollary}Typechecking of $\beta$-normal terms in System \fsubtop is decidable. 
\end{corollary}

\section{Conclusions and Further Directions}
We have described a semantics with two related interpretations of bounded quantification. Although this was presented via a simple syntactic reduction of bounded to unbounded quantification, its soundness depends fundamentally on a key semantic property,  \emph{dinaturality}, to relate subtype and parametric polymorphism, and arose from a more general investigation into the denotational semantics of bounded quantification.  

These semantic insights were applied to  give a type system which subsumes both Kernel \fsub and System \fsubtop\!. This sheds light on some of the troublesome aspects of the latter, in particular, by supplying  its missing minimal types. The price for this more well-behaved system --- having two forms of bounded quantification --- need not be paid by the programmer: by restricting to programs annotated with \fsubtop-types, we arrive at a system in which typechecking is decidable and the same set of $\beta$-normal forms can be typed as in System \fsubtop itself.

This treatment of bounded quantification is not dependent on the $\lambda$-calculus setting of System \fsubtop\!, and may transfer to related type systems such as the DOT calculus, where similar problems arise. Indeed, strong Kernel \fsub \cite{HuL}, which is similarly a fragment of System \fsub which achieves both decidability and greater expressiveness than Kernel \fsub\!,
 is derived from an analogous fragment of the type system ${\mathsf{D_{\leqt}}}$, which is the part of DOT without self-referencing and intersection types. Strong Kernel \fsub avoids the rebounding problem by deriving subtyping judgments with respect to two contexts, which may have different bounds for the same variable. A semantic account of this calculus (which is part of the broader aim to develop an intensional denotational semantics of object-oriented programming) may shed light on its expressiveness and relation to \fsubu\!.   

Another conclusion could be drawn from our semantic analysis: since meet types and dinaturality may be used to interpret both forms of bounded quantifier in terms of unbounded quantification, why not interpret programs directly in such a system (for which subtyping and typechecking are straightforward)? 

\bibliographystyle{entics}
\bibliography{names2}

\begin{thebibliography}{10}
\providecommand{\url}[1]{\texttt{#1}}
\providecommand{\urlprefix}{ }
\providecommand{\eprint}[2][]{\url{#2}}

\bibitem{Bainbridge}
Bainbridge, E.~S., P.~J. Freyd, A.~Scedrov and P.~Scott, \emph{Functorial
  polymorphism}, Theoretical Computer Science \textbf{70}, pages 35--64 (1990).
\newline\urlprefix\url{https://doi.org/10.1016/0304-3975(90)90055-m}

\bibitem{BruceLongo}
Bruce, K. and G.~Longo, \emph{A modest model of records, inheritance and
  bounded quantification}, Information and Computation \textbf{87}, pages
  196--240 (1990).
\newline\urlprefix\url{https://doi.org/10.1109/lics.1988.5099}

\bibitem{fsub}
Cardelli, L., J.~C. Mitchell, S.~Martini and A.~Scedrov, \emph{An extension of
  {S}ystem {F} with subtyping}, Information and Computation \textbf{109}, pages
  4--56 (1994).
\newline\urlprefix\url{https://doi.org/10.1006/inco.1994.1013}

\bibitem{CFun}
Cardelli, L. and P.~Wegner, \emph{On understanding types, data abstraction and
  polymorphism}, Computing Surveys \textbf{17}, pages 471 -- 522 (1985).
\newline\urlprefix\url{https://doi.org/10.1145/6041.6042}

\bibitem{fsubtop}
Castagna, G. and B.~C. Pierce, \emph{Decidable bounded quantification}, in:
  \emph{Proceedings of POPL '94}, pages 1--29 (1994).
\newline\urlprefix\url{https://doi.org/10.1145/174675.177844}

\bibitem{Chrob}
Chroboczek, J., \emph{Game semantics and subtyping}, in: \emph{Proceedings of
  the fifteenth annual symposium on Logic in Computer Science}, pages 192--203,
  IEEE press (2000).
\newline\urlprefix\url{https://doi.org/10.1109/lics.2000.855769}

\bibitem{CGhe}
Curien, P.-L. and G.~Ghelli, \emph{Coherence of subsumption, minimum typing and
  type-checking in $\mathsf{F_{<\! :}}$}, Mathematical Structures in Computer
  Science \textbf{2}, pages 55 -- 91 (1992).
\newline\urlprefix\url{https://doi.org/10.1007/3-540-52590-4_45}

\bibitem{delataillade}
de~Lataillade, J., \emph{Dinatural terms in {S}ystem {F}}, in:
  \emph{Proceedings of the 24th annual symposium on Logic in Computer Science,
  LICS '09}, IEEE Press (2009).
\newline\urlprefix\url{https://doi.org/10.1109/lics.2009.30}

\bibitem{HuL}
Hu, J. Z.~S. and O.~Lhot\'ak, \emph{Undecidability of {${\mathsf{D_{<\!:}}}$}
  and its decidable fragments}, Proceedings of the ACM on Programming Languages
  (POPL) \textbf{4}, pages 1--30 (2020).
\newline\urlprefix\url{https://doi.org/10.1145/3371077}

\bibitem{KaS}
Katiyar, D. and S.~Sankar, \emph{Completely bounded quantification is
  decidable}, in: \emph{Proceedings of the {ACM} SIGPLAN Workshop on ML and its
  Applications}, pages 68--77 (1992).
\newline\urlprefix\url{https://www.researchgate.net/publication/2763874_Completely_Bounded_Quantification_is_Decidable}

\bibitem{fos16}
Laird, J., \emph{Game semantics for bounded polymorphism}, in:
  \emph{Proceedings of FoSSaCS '16}, number 9634 in LNCS, Springer (2016).
\newline\urlprefix\url{https://doi.org/10.1007/978-3-662-49630-5_4}

\bibitem{dot1}
{N. Amin}, {S. Gr\"utter}, {M. Odersky}, {T. Rompf} and {S. Stucki}, \emph{The
  essence of dependent object types.}, in: \emph{A List of Successes That Can
  Change the World - Essays Dedicated to Philip Wadler on the Occasion of His
  60th Birthday}, number 9600 in LNCS, pages 249 -- 272, Springer (2016).
\newline\urlprefix\url{https://doi.org/10.1007/978-3-319-30936-1_14}

\bibitem{piercethesis}
Pierce, B.~C., \emph{Programming with Intersection Types and Bounded
  Polymorphism}, Ph.D. thesis, Carnegie Mellon University (1991).
\newline\urlprefix\url{https://doi.org/10.5555/145640}

\bibitem{Pierce1}
Pierce, B.~C., \emph{Bounded quantification is undecidable}, in: \emph{POPL},
  pages 305--315 (1992).
\newline\urlprefix\url{https://doi.org/10.1006/inco.1994.1055}

\bibitem{dot2}
Rompf, T. and N.~Amin, \emph{From {F} to {DOT:} type soundness proofs with
  definitional interpreters}, CoRR \textbf{abs/1510.05216} (2015).
  \eprint{1510.05216}.
\newline\urlprefix\url{http://arxiv.org/abs/1510.05216}

\bibitem{vor}
Vorobyov, S., \emph{Structural decidable extensions of bounded quantification},
  in: \emph{Proceedings of POPL '95}, pages 164--175 (1995).
\newline\urlprefix\url{https://doi.org/10.1145/199448.199479}

\end{thebibliography}

\end{document}